\documentclass [12pt]{article}

\usepackage{fullpage,wrapfig}
\usepackage{amsmath, amsthm, amssymb, algorithm, thmtools, algpseudocode, enumerate, caption, framed}
\usepackage{paralist}

\usepackage[usenames,dvipsnames,svgnames,table]{xcolor}
\definecolor{darkgreen}{rgb}{0.0,0,0.9}

\usepackage[colorlinks=true,pdfpagemode=UseNone,citecolor=Blue,linkcolor=Red,urlcolor=BrickRed,
pagebackref]{hyperref}

\usepackage[T1]{fontenc}
\usepackage[utf8]{inputenc}
\usepackage{authblk}

\usepackage{tikz}

\numberwithin{equation}{section}

\newtheorem{theorem}{Theorem}[section]

\newtheorem{lemma}{Lemma}[section]

\newtheorem{corollary}[theorem]{Corollary}

\date{}

\title{On Guarding Orthogonal Polygons with Sliding Cameras\footnote{Research of TB and TC supported by NSERC. Research of FM supported in part by the MIUR project AMANDA, prot. 2012C4E3KT\_001.  Research was done while FM was visiting the University of Waterloo.}}

\author[1]{Therese Biedl}
\author[1]{Timothy M. Chan}
\author[1]{Stephanie Lee}
\author[1]{Saeed Mehrabi}
\author[2]{Fabrizio Montecchiani}
\author[1]{Hamideh Vosoughpour}
\affil[1]{{\small David R. Cheriton School of Computer Science, University of Waterloo, Waterloo, Canada.
	\texttt{\{biedl, tmchan\}@cs.uwaterloo.ca, \{s363lee, smehrabi, hvosough\}@uwaterloo.ca}}}
\affil[2]{{\small Department of Engineering, University of Perugia, Perugia, Italy.
	\texttt{fabrizio.montecchiani@unipg.it}}}

\begin{document}

\maketitle

\begin{abstract}
A sliding camera inside an orthogonal polygon $P$ is a point guard that travels back and forth along an orthogonal line segment $\gamma$ in $P$. The sliding camera $g$ can see a point $p$ in $P$ if the perpendicular from $p$ onto $\gamma$ is inside $P$. 
In this paper, we give the first constant-factor approximation algorithm for the problem of guarding $P$ with the minimum number of sliding cameras.   Next, we show that the sliding guards problem is linear-time solvable if the (suitably defined) dual graph of the polygon has bounded treewidth.  Finally, we study art gallery theorems for sliding cameras, thus, give upper and lower bounds in terms of the number of guards needed relative to the number of vertices $n$.
\end{abstract}

\section{Introduction}
\label{sec:intro}
Let $P$ be a (not necessarily orthogonal) polygon with $n$ vertices. The art gallery problem, posed by Victor Klee in 1973~\cite{orourke1982}, asks for the minimum number of point guards required to guard $P$, where a point guard $g$ sees a point $p\in P$ if the line segment connecting $g$ to $p$ lies inside $P$. Chv\'{a}tal~\cite{chvatal1975} was the first to answer the question by giving the tight bound $\lfloor n/3 \rfloor$ on the number of point guards that are needed to guard a simple polygon with $n$ vertices. For polygons with holes, Hoffmann et al.~\cite{hoffmannFOCS1991} proved that $\lfloor (n+h)/3 \rfloor$ point guards are always sufficient and occasionally necessary, where $h$ is the number of holes. For \emph{orthogonal polygons}, it was proved multiple times \cite{jeff1983, Lubiw1985, orourke1982} that $\lfloor n/4 \rfloor$ point guards are always sufficient and sometimes necessary to guard the interior of a simple orthogonal polygon with $n$ vertices. 

Finding the minimum number of guards
is \textsc{NP}-hard on simple polygons
\cite{lee1986}, even on simple orthogonal polygons \cite{dietmar1995} 
or monotone polygons~\cite{krohn2013}.  A number of results concerning
approximation algorithms are also known
\cite{eidenbenz2001, Kirkpatrick2015, krohn2013}. 

\noindent\paragraph{Mobile Guards and Sliding Cameras.} A \emph{mobile guard} is a point guard that travels along a line segment $\gamma$ inside $P$. Guard $\gamma$ can see a point $p$ in $P$ if and only if there exists a point $g\in \gamma$ such that the line segment $pg$ lies entirely inside $P$. If the line segment $\gamma$ must be orthogonal, then we call it an \emph{orthogonal mobile guard}. Moreover, if the line segment $pg$ is required to be perpendicular to $\gamma$, then we call $\gamma$ a \emph{sliding camera}. Note that an orthogonal mobile guard travelling along $\gamma$ may see a larger area of $P$ than a sliding camera travelling along $\gamma$, see also Figure~\ref{fig:pWithHoles}(a).
\begin{figure}[t]
\centering
\hspace*{\fill}
\includegraphics[width=.40\textwidth]{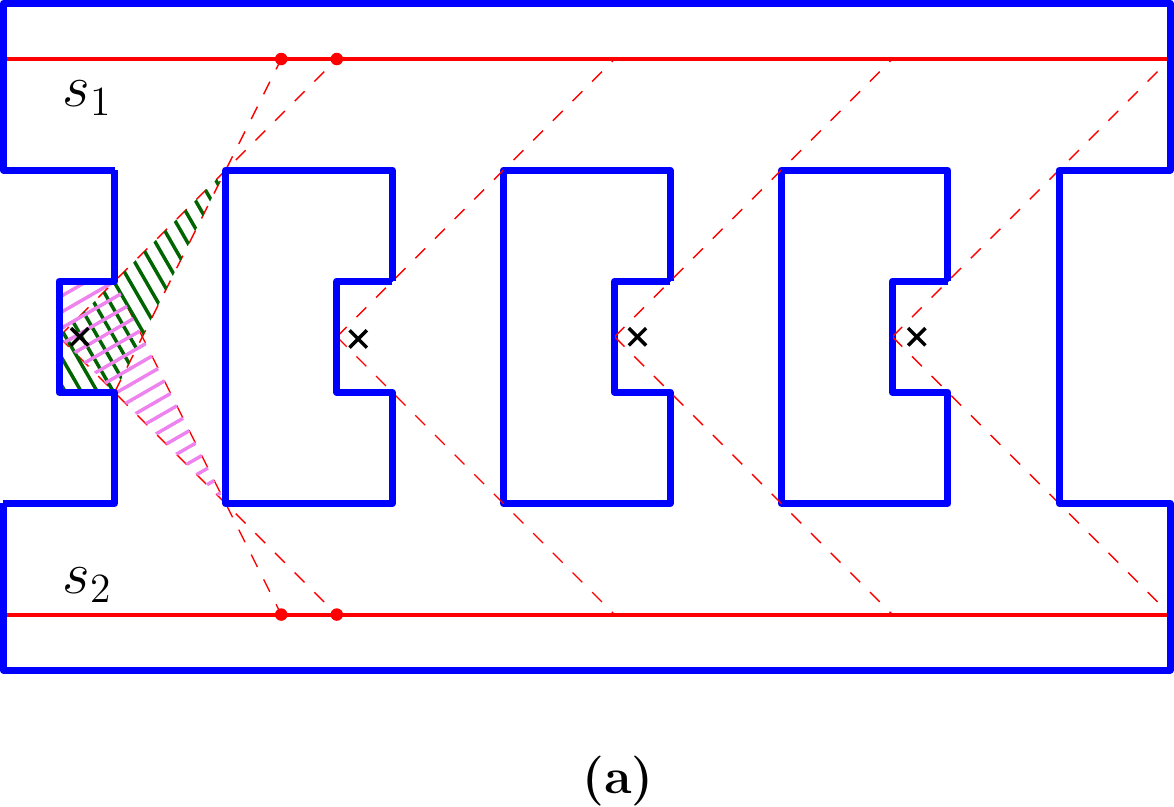}
\hspace*{\fill}
\includegraphics[width=.40\textwidth]{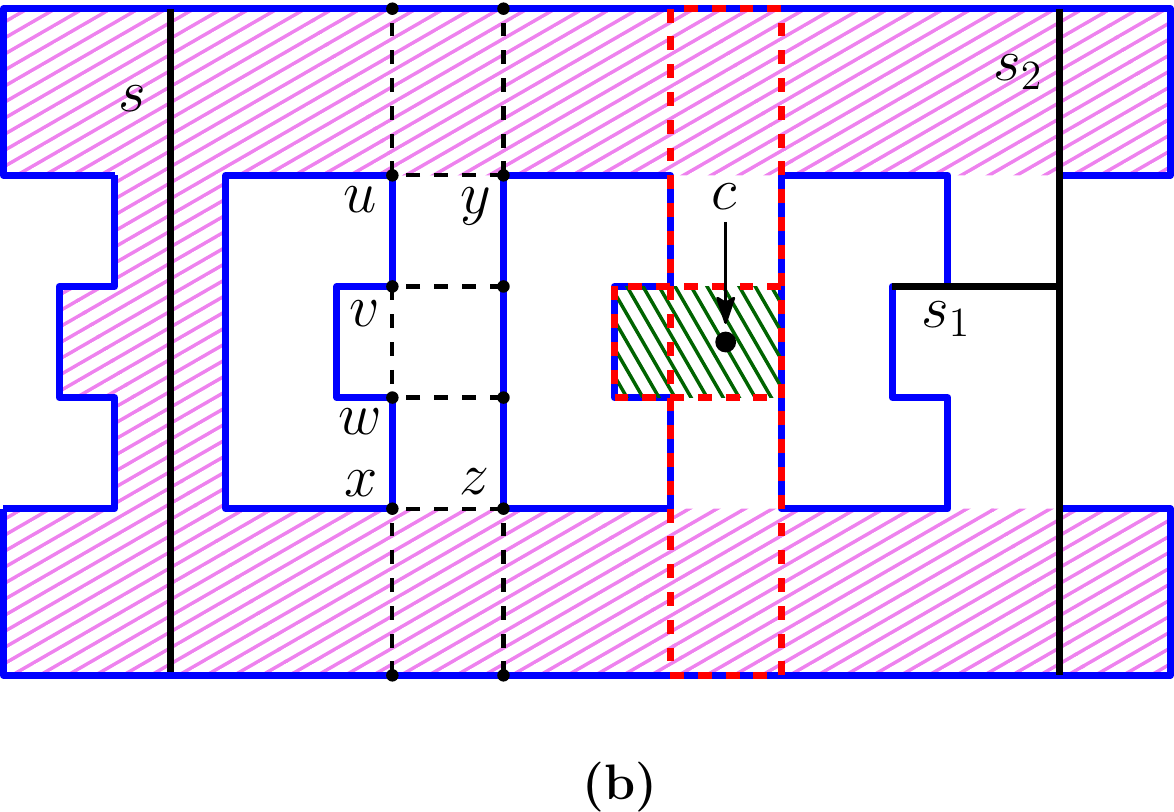}
\hspace*{\fill}
\caption{(a) An orthogonal polygon $P$ that can be guarded with two orthogonal mobile guards, but requires $\Theta(n)$ sliding cameras to be guarded since no two crosses can be seen by one sliding camera. (b) Sliding camera $s$ sees the rising-shaded subpolygon of $P$.  We also show parts of the pixelation induced by rays from reflex vertices $\{u, v, w, x, y, z\}$, and the cross $c$ whose supporting horizontal slices is downward shaded. Segments $s_1$ and $s_2$ are guard-segments.}
\label{fig:pWithHoles}%
\end{figure}
The notion of mobile guards was introduced by Avis and Toussaint~\cite{avis1981}. O'Rourke~\cite{orourke1983} proved that $\lfloor n/4\rfloor$ (not necessarily orthogonal) mobile guards are sufficient for guarding arbitrary polygons with $n$ vertices. For orthogonal polygons with $n$ vertices, $\lfloor (3n+4)/16\rfloor$ mobile guards are always sufficient and sometimes necessary~\cite{aggarwal1984}.

In this paper we study the \emph{Minimum Sliding Cameras (MSC)} problem, i.e., we want to guard an orthogonal polygon $P$ with the minimum number of sliding cameras.  We also consider the variant \emph{Minimum Horizontal Sliding Cameras (MHSC)} where only horizontal cameras are allowed.
These problems were introduced by Katz and Morgenstern~\cite{katz2008}, who proved that MHSC can be solved in polynomial time in the special case where the polygon is simple (has no holes).  It was shown later that MSC is NP-hard in polygons with holes \cite{durocher2013}; NP-hardness in simple polygons is open.  Durocher et al.~\cite{durocherLATIN2014} claimed a (3.5)-approximation algorithm for MSC problem on simple orthogonal polygons, but this was later discovered by the authors to be incorrect (private communication). For the special case of monotone orthogonal polygons, Katz and Morgenstern~\cite{katz2008} gave a 2-approximation algorithm, which was later improved by De Berg et al.~\cite{deBergCOCOA2014} to a linear-time exact algorithm.

\noindent\paragraph{Our Results.} In this paper, we give hardness results and algorithms for both MSC and MHSC.    
Specifically, we give two (conceptually very different) algorithms.  The first works by constructing a small $\varepsilon$-net for the hitting set problem that naturally arises from MSC.  This gives then an $O(1)$-approximation algorithm for the MSC problem on orthogonal polygons. Note that no constant-factor approximation algorithm was known previously,
and, as opposed to previous attempts at such approximation-algorithms \cite{durocherLATIN2014}, our algorithm works even on orthogonal polygons with holes.  
The second algorithm uses a tree-decomposition approach.  We show that if the dual graph of the so-called pixelation of the polygon has bounded treewidth,
then MSC can be solved in polynomial time.  In particular, MSC becomes polynomial in so-called thin polygons that have no holes.

Both the above approaches also work (and become even simpler) for MHSC where only horizontal cameras are allowed.
We also establish NP-hardness of MHSC for polygons with holes.  The same proof also works for MSC and is different, and perhaps simpler,
than the previous NP-hardness proof for MSC~\cite{durocher2013}.

Finally, we consider art gallery theorems for sliding cameras, i.e.,  theorems that bound the number of guards relative to the number of vertices.  We present the following results for an orthogonal polygon $P$ with $n$ vertices: \begin{inparaenum}[(i)] \item $\lfloor (3n+4)/16 \rfloor$ sliding cameras are always sufficient and sometimes necessary to guard $P$ entirely, \item if the dual graph induced by the vertical decomposition of $P$ is a path, then $\lfloor (n+2)/6\rfloor$ sliding cameras are always sufficient to guard $P$, and \item $\lfloor n/4\rfloor$ horizontal sliding cameras are always sufficient and sometimes necessary to guard $P$ entirely. \end{inparaenum}

\section{Preliminaries}
\label{sec:prelim}
Throughout the paper, $P$ denotes an orthogonal polygon with $n$ vertices.
The {\em horizontal} (respectively {\em vertical}) {\em segmentation} of
$P$ consists of extending a horizontal (vertical) ray inward from any
reflex vertex of $P$ until it hits another vertex or edge.  The rectangles
in the resulting partition of $P$ are called the horizontal (vertical)
{\em slices} of $P$.  Each slice can be represented by the horizontal
(vertical) line segment that halves the slice;
we call these the {\em slice-segments} and denote them by $\Sigma$.

The {\em pixelation} of $P$ is obtained by doing both the horizontal and
the vertical segmentation of $P$.  The resulting rectangles are called
{\em pixels}. 
The pixelation may well have $\Theta(n^2)$ pixels.
Notice that the pixels are in 1-1-correspondence with pairs of slices that 
cross.  We can hence identify each pixel with
a {\em cross} $c$, which is the point where the two slice-segments 
$\sigma_H$ and $\sigma_V$ of
these two slices cross.  We say that $\sigma_H$ and $\sigma_V$ {\em support}
$c$.  Denote the set of crosses by $X$.

A {\em sliding camera} $\gamma$ is a horizontal or vertical line segment that 
is inside $P$.  (We will frequently omit ``sliding'', as we study no other
type of camera.)  The region visible from $\gamma$ is the set of all points $p$
such that that perpendicular from $p$ to $\gamma$ is inside $P$.
Note that doing a {\em parallel shift} (i.e., translating a horizontal camera 
vertically or a vertical camera horizontally)
does not change its visibility region for as long as we stay inside $P$. We
may hence assume that any camera runs along pixel-edges.
We may also restrict our attention to cameras that are maximal line segments 
within $P$ (all others would see a subset).  Let $\Gamma$
be the set of {\em guard-segments} which are maximal horizontal and vertical
line  segments within $P$ that run along pixel edges.
See Figure~\ref{fig:pWithHoles}(b) for an illustration.

The following lemma is a simple re-formulation of what guarding means,
but casts the problem into a discrete framework that will be crucial later.

\begin{lemma}
\label{lem:hitting}
A set $S$ of $k$ sliding cameras guards polygon $P$ if and only if
there exists a set of $k$ guard-segments 
$S'\subseteq \Gamma$ such that for every cross $c\in X$, at least one of the slice-segments that support $c$ is intersected by some $\gamma\in S'$.
\end{lemma}
\begin{proof}
First assume that we have such a set $S'\subseteq \Gamma$.  Any point
$p$ belongs to some pixel $\psi$, which corresponds to a cross $c$.
By assumption one 
guard $\gamma\in S'$ intersects one slice-segment $\sigma$ that supports
$c$.  Since $\gamma$ is maximal, it therefore sees the entire slice
corresponding to $\sigma$, therefore all of $\psi$, and therefore $p$.
So $S'$ is a set of sliding cameras that guard all of $P$.

Vice versa, assume we have a set $S$ of sliding cameras that guard all 
of $P$.  
For each camera, do a parallel translation until it hits the boundary of $P$,
then extend it until it is a maximal line segment within $P$; this
can only increase its visibility region.  
The resulting set $S'$ of cameras now are guard-segments that
guard $P$.  Now fix one cross $c$, which is
a point inside $P$ and hence is guarded by some $\gamma\in S'$.  This means
that for some $g\in \gamma$ the line segment $gc$ is perpendicular to 
$\gamma$ and inside $P$.  But $gc$ is part of one slice-segment that
supports $c$, and so $g$ is the intersection point between that slice-segment
and guard-segment $\gamma$, which proves the claim.
\end{proof}

We say that a guard-segment $\gamma$ 
\emph{hits} a cross $c$ if and only if $\gamma$ intersects one of the
slice-segments supporting $c$.  Lemma~\ref{lem:hitting} can then be re-stated 
as that $S'$ hits all crosses.    In fact,
the algorithms we design later will allow further restrictions:
we can specify exactly which crosses should be hit and which cameras 
may be used as guards.  So assume we are given some $X'\subseteq X$ and some 
$\Gamma'\subseteq \Gamma$.  The {\em $(X',\Gamma')$-sliding cameras problem} consists of 
finding a minimum subset of cameras in $\Gamma'$ that hit all crosses in $X'$.  
Note that with a suitable choice of $\Gamma'$ this encompasses both MSC and MHSC.

\section{Approximation Algorithms via $\varepsilon$-Nets}
\label{sec:eNet}

In this section, we give approximation algorithms for MSC and MHSC that are
based on phrasing the problem as a hitting set problem and then using
$\varepsilon$-nets.  We do this first for MHSC, and then later
re-use those $\varepsilon$-nets for MSC.

\noindent\paragraph{Hitting Sets.}  
A \emph{set system} is a pair $\mathcal{R}=(\mathcal{U}, \mathcal{S})$, 
where $\mathcal{U}$ is a universe set of objects and $\mathcal{S}$ is a collection of 
subsets of $\mathcal{U}$. A \emph{hitting set} for the set system $(\mathcal{U}, \mathcal{S})$ 
is a subset of $\mathcal{U}$ that intersects every set in $\mathcal{S}$.  

For the $(X',\Gamma')$-sliding camera problem, we construct a set system as follows. 
Let $\mathcal{U}=\Gamma'$ be all potential sliding cameras.  For each 
cross $c\in X'$ that needs to be hit, define $S_c$ to be all the cameras in $\mathcal{U}$ 
that hit $c$, and let $\mathcal{S}$ be the collection of these sets.
From the definitions, finding a hitting set for this set system is the same as solving the $(X',\Gamma')$-sliding-camera problem.

An {\em $\varepsilon$-net} for a set system $\mathcal{R}=(\mathcal{U}, \mathcal{S})$ is a subset $N$ of $\mathcal{U}$ such that every set $S$ in $\mathcal{S}$ with size at least $\varepsilon\cdot\lvert\mathcal{U}\rvert$ has a non-empty intersection with $N$. Br\"{o}nnimann and Goodrich~\cite{BronnimannG95} showed that $\varepsilon$-nets can be used to derive approximation algorithms as follows.  Define a \emph{net finder} to be a (poly-time) algorithm that, for a given set system $\mathcal{R}=(\mathcal{U}, \mathcal{S})$ and any given $r>0$, computes an $(1/r)$-net of $\mathcal{R}$ whose size is at most $s(r)$ for some function $s$.  Also, a {\em verifier} is a poly-time algorithm that, given a subset $H\subset \mathcal{U}$, states (correctly) that $H$ is a hitting set, or returns a non-empty set $R\in \mathcal{S}$ such $H$ does not hit $S$.

\begin{lemma}\cite{BronnimannG95}
Let $\mathcal{R}$ be a set system that admits both a poly-time net finder and a poly-time verifier. Then there is a poly-time algorithm that computes a hitting set
of size at most $s(4 \cdot \text{OPT})$, where OPT stands for the size of an optimal hitting set, and $s(r)$ is the size of the $(1/r)$-net.
\end{lemma}

Thus, the lemma gives an $O(1)$-approximation algorithm for as long as we can find an $\varepsilon$-net
whose size is $O(1/\varepsilon)$.
(Clearly the hitting set problems defined by MHSC and MSC both have a 
polynomial-time verifier.)  

\noindent\paragraph{An $\varepsilon$-net for the MHSC Problem.}
We now show the existence of such a small
$\varepsilon$-net for MHSC.  For this, we need (yet another) 
reformulation that simplifies the problem.

\begin{lemma}
A set $S$ of horizontal guard-segments hits all crosses in a set $U'$ 
if and only if $S$ intersects all the vertical slice-segments that support
crosses in $U'$.
\end{lemma}
\begin{proof}
If camera $\gamma$ hits cross $c$, then it intersects either its
horizontal supporting slice-segment $\sigma_H$ or its vertical supporting
slice-segment $\sigma_V$. 
But if $\gamma$ intersects $\sigma_H$, then since both are horizontal
and $\gamma$ is maximal we have $\sigma_H\subseteq \gamma$,
in case of which $\gamma$ 
also contains point $c$ and therefore intersects $\sigma_V$.
So either way $\gamma$ intersects $\sigma_V$.
\end{proof}

For MHSC, it hence suffices to represent every cross by its vertical
slice-segment and so reduce the problem to the following:
Given a set of horizontal line segments ${\cal H}$
and a set of vertical line segments ${\cal V}$, find a minimum set
$S\subseteq {\cal H}$ such that every line segment in $\mathcal{V}$ is
intersected by $S$.  This problem is also known as the {\em Orthogonal
Segment Covering} problem and was shown to be NP-complete \cite{KatzMN2005}.
We hence have:

\begin{corollary}
MHSC reduces to the Orthogonal Segment Covering problem.
\end{corollary}

In the following, we show that the Orthogonal Segment Covering problem
has a small  $\varepsilon$-net; by the above this immediately implies a
small $\varepsilon$-net for the hitting set problem for MHSC.
 
\begin{lemma}
The Orthogonal Segment Covering problem has a $(1/r)$-finder
with size-function $s(r)\in O(r)$.
\end{lemma}
\begin{proof}
We employ a result of
Clarkson and Varadarajan~\cite{clarksonV2007} that shows that 
$\varepsilon$-nets of small size can be found for hitting set problems
in geometric objects, using random sampling.  The size is as desired
under the assumption that these geometric objects have small union
complexity, i.e., the union of $n$ of these geometric objects has 
complexity $O(n)$.

Observe that a horizontal line segment
$\gamma=[x, x'] \times  y$ intersects
a vertical 
line segment $\sigma=a \times [b, b']$ if and only if the point $q=(x, y, x')$ 
lies in the range $Q=(-\infty, a] \times [b, b'] \times [a, \infty)$. Hence, 
we map every horizontal segment into a point and every vertical segment
into an extruded translated quadrant so that the segments intersect if and only
if the point lies in the extruded translated quadrant.

Viewing the $y$-coordinate as \emph{time}, each vertical segment hence
is mapped to a translated quadrant that exists for some time, and the union
complexity becomes the number of changes to a dynamic two-dimensional 
union of translated quadrants.  In general there could be a quadratic
number of such changes, but our translated quadrants are special: the
corner is on the line $\{x=z\}$.    As such, for any two such translated
quadrants one contains the other one, which means that adding or deleting
one such quadrant to the union causes a constant number of changes.
Hence the union complexity of our extruded translated quadrants is
linear. By \cite{clarksonV2007} hence there exists a $(1/r)$-net
with size $O(r)$ for hitting these extruded translated quadrants
by points, hence hitting vertical segments by horizontal segments.
\end{proof}

Combining the above results gives:

\begin{theorem}
There exists a poly-time  $O(1)$-approximation algorithm for the Orthogonal
Segment Covering problem and the MHSC problem.
\end{theorem}

\noindent\paragraph{An $\varepsilon$-net for the MSC Problem.}  Using the $\varepsilon$-net
for MHSC, we can easily find one for MSC and hence have an approximation
algorithm for this as well.

\begin{theorem}
There exists a poly-time $O(1)$-approximation algorithm for the MSC problem.
\end{theorem}
\begin{proof}
Fix a polygon $P$ and consider the
$(X',\Gamma')$-sliding camera problem for $P$.
It suffices to show that for any $r>0$ there exists a $1/r$-net $T$ 
of size $O(r)$ for the corresponding hitting set problem $\mathcal{R}$.
Let $T_H$ be a $1/2r$-net for the hitting
set of MHSC for $P,X'$ and the horizontal cameras in $\Gamma'$.
Let $T_V$ be a $1/2r$-net for the
hitting set of MVSC (i.e., when we want to guard the polygon using only
vertical sliding guards) for $P,X'$ and the vertical cameras in $\Gamma'$. 
Set $T:=T_H\cup T_V$.
We claim that $T$ is a $1/r$-net for $\mathcal{R}$. 

So assume some set $S_c$ in the hitting-set problem satisfies 
$|S_c|\geq |{\cal U}|/r$.  Translating back, this means that
some cross $c\in X'$ is hit by at least $|\Gamma'|/r$ guard-segments. 
Assume w.l.o.g.~that at least half of these hitting guard-segments are
horizontal.  
Then the vertical slice-segment $\sigma_V$ that supports $c$ intersects
at least $|\Gamma'|/2r$ horizontal guard-segments in $\Gamma'$.  
By definition of a $(1/2r)$-net, therefore there is a line segment 
$\gamma\in T_H$ that intersects $\sigma_V$.
Therefore $\gamma\in T$ hits $c$ as required.
\end{proof}

\section{Polygons Whose Pixelation Has Bounded Treewidth}
\label{sec:treedecomp}

Recall that the {\em pixelation} of a polygon is obtained by cutting
the polygon horizontally and vertically at all reflex vertices.
The {\em dual graph $D$ of the pixelation} is obtained by interpreting
the pixelation as a planar graph and taking its weak dual (i.e., dual
graph but omit the outer face).
Thus, $D$ has a vertex
for every pixel of $P$, and two pixels are adjacent in $D$ if and only 
if they share a side.    We now show how to solve MSC and MHSC under
the assumption that $D$ has small treewidth.  (Our approach was inspired
by a similar result for a different guarding problem \cite{otherESA}, but 
the construction here is simpler.)

\noindent\paragraph{2-dominating Set in an Auxiliary Graph.}
By Lemma~\ref{lem:hitting}, the $(X',\Gamma')$-sliding camera problem is 
equivalent to finding a set of guard-segments
that hits at least one supporting slice-segment of each cross.  
This naturally gives rise to an auxiliary graph $H$ as follows:
The vertices of $H$ are $X'\cup \Sigma\cup \Gamma'$.
For any $c\in X'$, add an edge from $c$ to each of its two supporting
slice-segments.
For any guard-segment $\gamma$, add an edge to 
any slice-segment that it intersects.
From Lemma~\ref{lem:hitting}, and since there are no edges from $X'$ to $\Gamma'$, one immediately sees the following:

\begin{lemma}
The minimum guard set for the $(X',\Gamma')$-sliding-cameras problem corresponds to a subset 
$S\subseteq \Gamma'$ of vertices in $H$ such that all vertices in $X'$ are within distance 2 from $S$.
\end{lemma}

The above lemma means that the sliding-camera problem reduces 
to a graph-theoretic problem that is quite similar to the 2-dominating set
(the problem of finding a set $S$ such that all other vertices have
distance at most 2 from $S$); the only change is that we restrict
which vertices may be used for $S$ and which vertices must be within
distance 2 from $S$.  2-dominating set is an NP-hard problem in general,
but is easily shown to be polynomial in graphs that have bounded treewidth,
which we define next.

\noindent\paragraph{Treewidth.} A {\em tree decomposition} ${\cal T}=(I,{\cal X})$
of a graph $G=(V,E)$ consists of a tree $I$ and an assignment 
${\cal X}:V(I)\rightarrow 2^V$ of \emph{bags} of vertices
of $G$ to nodes of $I$
 such that the following holds: (a) for every vertex $v\in V$,
the set of bags containing $v$ forms a non-empty connected subtree of $I$,
(b) for every edge $e\in E$, at least one bag contains both ends of $e$.
The {\em width} of a tree decomposition is the maximum bag-size minus 1,
and the {\em treewidth} $tw(G)$ of a graph $G$ is the smallest possible width  
over all tree decompositions of $G$.    In particular, a tree has treewidth 1.

\begin{lemma}
Let $P$ be a polygon whose dual graph $D$ of the pixelation has treewidth at 
most $k$.
Then for any choice of $X'\subseteq X$ and $\Gamma' \subseteq \Gamma$,
the auxiliary graph $H$ has treewidth at most $7k+6$.
\end{lemma}
\begin{proof}
Let ${\cal T}=(I,{\cal X})$ be a tree decomposition of $D$ of width $k$.
Any bag ${\cal X}(i)$ hence contains at most $k+1$ nodes of $D$.  Any node
of $D$ corresponds to a pixel $\psi$ of $P$, hence a cross $c$ in $X$.  Let
$\sigma_H$ and $\sigma_V$ be the two supporting slice-segments of $c$.  Also,
let $\gamma_N,\gamma_E,\gamma_S$ and $\gamma_W$ be the four guard-segments 
that run along the four sides of pixel $\psi$.
Define a tree decomposition for $H$
by using the same tree $I$ and defining a new bag ${\cal X}'(i)$ 
by replacing each pixel $\psi$ by the 7 items
$c,\sigma_H,\sigma_V,\gamma_N,\gamma_E,\gamma_S$ and $\gamma_W$.
(We can omit $c$ if it is not in $X'$ and any $\gamma$ that is not in $\Gamma'$.)
Clearly the new bag has size at most $7(k+1)$,
and so the width of the tree decomposition is at most $7k+6$.

It remains to argue that this is indeed a tree decomposition of $H$.
Clearly the subtree of bags that contain cross $c\in X'$ is connected,
since these are the same bags as the ones in ${\cal T}$ that contained
the corresponding pixel.
Any slice-segment $\sigma$ occurs in all bags of some cross supported by $\sigma$.
But the pixels of these crosses form a connected subgraph of $D$ (they
are connected along $\sigma$),
and for any connected subgraph the bags containing vertices of it
are connected.  Hence $\sigma$ occurs in a connected subtree of bags.
Likewise any guard $\gamma$ occurs in exactly those bags that contained
a pixel where $\gamma$ occupies one of the sides.  The subgraph of these
pixels is connected,  and so the bags containing $\gamma$ are also connected.

To verify the edge-condition, observe that if $(c,\sigma)$ is an edge from
a cross $c$ to slice-segment $\sigma$, then any bag containing $c$ also contains 
$\sigma$. Next, if $(\sigma,\gamma)$
is an edge for some slice-segment $\sigma$ and guard-segment $\gamma$, then let $x$ 
be the point where $\sigma$ and $\gamma$ intersect, and let $\psi$ be
a pixel containing $x$.  Then any bag that used to contain $\psi$ now
has $\sigma$ and $\gamma$ in it.
Hence all
edges are covered and we have constructed a tree decomposition of $H$.
\end{proof}

To apply this treewidth-result, we must show that the
problem can be expressed as a suitable logic-formula.
(See e.g.~\cite[Chapter 7.4]{cygan2015} for more details.)
In particular, the following formula will do:  A set $S$ of guard-segments
that guards $X'$ satisfies
$$ S\subseteq \Gamma' \quad \land \quad \forall u\in X' \left( 
	\exists \sigma\in \Sigma \quad \text{adj}(u,\sigma) \land 
	\exists \gamma\in S \quad \text{adj}(\sigma,\gamma) \right) $$
(where $\text{adj}$ is a logic-formula to encode that the two parameters
are adjacent in $H$).  Since $H$ has bounded treewidth, we can find
the smallest set $S$ that satisfies this (or report that no such $S$
exists if $\Gamma'$ was too small) in linear time using Courcelle's
theorem \cite{Courcelle1990}.  Putting everything together, we hence have:

\begin{theorem}
If $P$ is a polygon whose dual graph has bounded treewidth, then
the $(X',\Gamma')$-sliding-cameras problem can be solved in linear time.
\end{theorem}

We give one application of this result.  A {\em thin} polygon
is a polygon for which no pixel-corner is in the interior.  MSC and MHSC
are NP-hard even for thin polygons with holes (as we will see in the
next subsection).  However, for thin polygons without holes, the dual
graph of the pixelation is clearly a tree, hence has bounded treewidth,
and both MSC and MHSC can be solved in linear time.

\begin{corollary}
If $P$ is a thin polygon without holes, then MSC and MHSC can be solved in linear time.
\end{corollary}

This result is not directly comparable to existing results~\cite{katz2008, deBergCOCOA2014}: it is stronger than these since it
works for MSC and does not require monotonicity, but it is weaker than
these since it requires a thin polygon.

A natural question is whether this result for bounded treewidth could
be used to generate a PTAS, by splitting the polygon (hence the planar graph) 
suitably and applying the ``shifting technique'' (see
\cite{baker1994} or \cite[Chapter 7.7.3]{cygan2015}).
We have not been able to develop such a PTAS, principally because
the cameras are not ``local'' in the sense that they can guard pixels
that have arbitrarily large distance in $D$.  Creating a PTAS (or proving 
APX-hardness) hence remains an open problem.

\section{NP-Hardness of MHSC}
\label{sec:npHardness}
Recall that the MHSC problem is polynomial-time solvable on simple orthogonal polygons~\cite{katz2011}. 
We show in this section that the problem becomes \textsc{NP}-hard on orthogonal polygons with holes. 
We note that the hardness proof of Durocher and Mehrabi~\cite{durocher2013} does not apply to the MHSC 
problem because they require both horizontal and vertical sliding cameras. 

The reader may recall that we showed that MHSC reduces to the
Orthogonal Segment Covering problem, which is known to be NP-hard 
\cite{KatzMN2005}.
However, this does not prove NP-hardness of MHSC, because not every instance
of Orthogonal Segment Covering can be expressed as MHSC.  Instead we give
a different reduction from Minimum Vertex Cover on max-deg-3 planar graphs.  
This problem (which is \textsc{NP}-hard~\cite{GareyJ1977}) 
consists of, given a planar graph $G=(V,E)$ 
with at most 3 incident edges at each vertex, find a minimum set 
$C\subseteq V$ such that for every edge at least
one endpoint is in $C$.  

Given a max-deg-3 planar graph $G$, we first compute a \emph{bar 
visibility representation} of $G$, that is, we assign to each vertex a
horizontal line segment (called {\em bar}) and to each edge $(v,w)$ a vertical 
{\em strip} of positive width that joins the corresponding bars and that is
disjoint from all other strips.
It is well-known that every planar graph has such a representation (see 
e.g.~Tamassia and Tollis~\cite{TamassiaT1986}), and it can be found in linear time.
By making strips sufficiently thin, we can ensure that no two strips of edges
occupy the same $x$-range.

From this visibility representation,
we construct an orthogonal polygon $P$ such that $G$ has a vertex cover of size $k$ if and only if $P$ can be guarded by $3N+k$ horizontal sliding cameras, where $N$ is the number of the vertices of $G$.  $P$ will have a hole for every
inner face of $G$ (and perhaps some additional holes).

An initial version of the polygon consists of the visibility representation,
with bars thickened into positive height.  Clearly if we have a vertex cover
of $G$, then we can place horizontal cameras in the bars of all corresponding
vertices and guard all edge-strips.  This is, however, not sufficient to also
guard all vertex-bars, which is why we add small gadgets that require
additional cameras, and these cameras guard none of the edge strips but all
other parts of the vertex bars.

We add three gadgets, called \emph{elephant gadgets}, around every vertex-bar $u$ in $\Gamma$. This might be done in one of the three different ways, up to symmetry, depending on how the edge-strips are connected to $u$; see Figure~\ref{fig:elephantGadgets} for an illustration.   Recall that we assume that no two edge-strips share an $x$-range, so that one of these constructions is always possible if the vertex has degree 3.  For vertices of smaller degrees, we omit some of the edge-strips at the ends of the bars (but not the elephant-gadgets) appropriately.

\begin{figure}[t]
\centering
\hspace*{\fill}
\includegraphics[page=1,width=.30\textwidth]{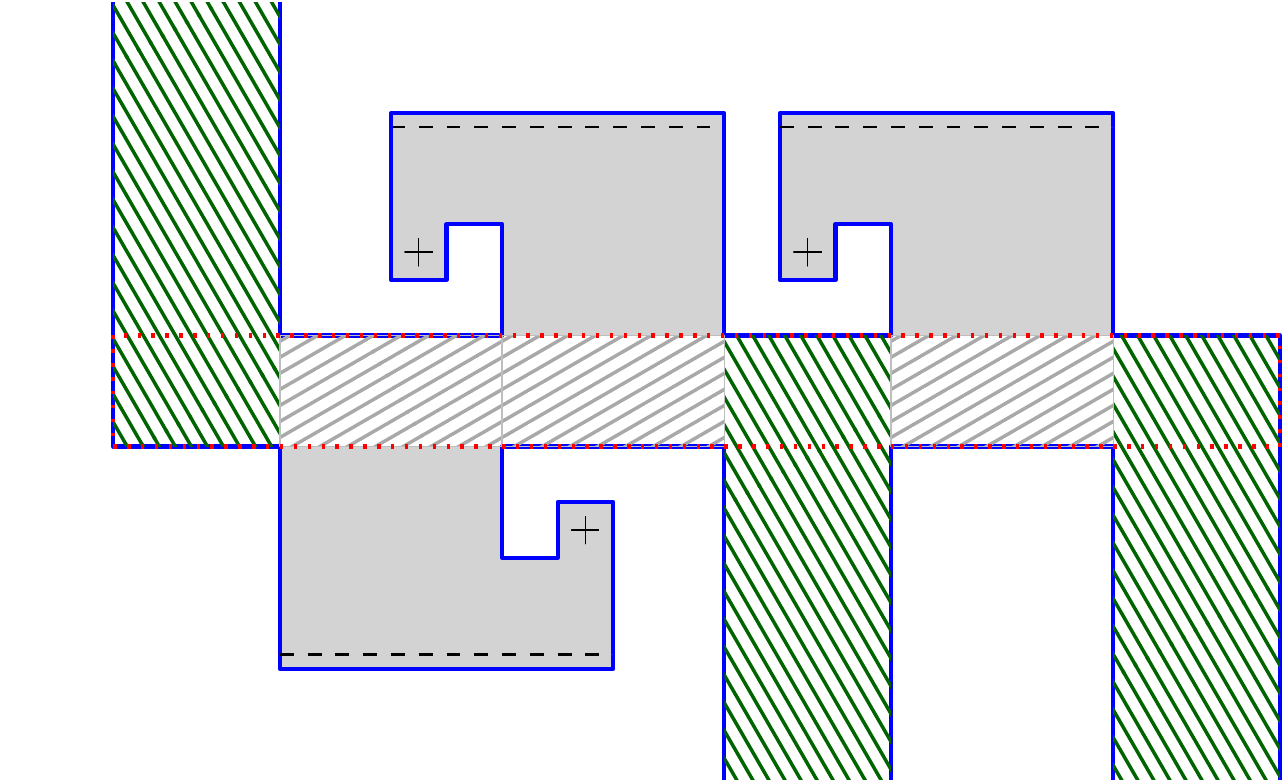}
\hspace*{\fill}
\includegraphics[page=2,width=.30\textwidth]{elephantsNew.pdf}
\hspace*{\fill}
\includegraphics[page=3,width=.30\textwidth]{elephantsNew.pdf}
\hspace*{\fill}
\caption{The NP-hardness construction.  Vertex-bars are red (dotted).  
Edge-strips are green (falling pattern).  Elephant-gadgets are gray; the dashed segment
indicates the spine-camera and the cross indicates the trunk.   Gray
(rising pattern) regions indicates the part of the vertex-box guarded by the
spine-cameras.}
\label{fig:elephantGadgets}%
\end{figure}

\begin{lemma}
\label{lem:onlyHSCReduction}
$G$ has a vertex cover of size $k$ if and only if $P$ can be guarded with $k+3N$ horizontal sliding cameras if and only if $P$ can be guarded with $k+3N$ sliding cameras.
\end{lemma}
\begin{proof}
Given a vertex cover $C$ of size $k$, we can place one  
horizontal ``vertex-camera'' inside the bar of each vertex of $C$, and 
one horizontal camera ``spine-camera'' inside each elephant-gadget. 
Then the vertex-cameras 
cover all edge-strips, and the spine-cameras cover all 
the elephants and the parts of vertex-bars not in edge-strips.  So we
can guard $P$ with $k+3N$ horizontal cameras.  

Clearly, if we can guard $P$ with $k+3N$ horizontal cameras, then it can
also be guarded with $k+3N$ arbitrary cameras.

Finally assume we can guard $P$ with $k+3N$ cameras (not necessarily horizontal
ones).  Define a vertex set $C$ as follows:  If any horizontal 
camera intersects the bar of vertex $v$, then add $v$ to $C$.  If any
vertical guard intersects the edge-strip of $(v,w)$, then arbitrarily
add one of $v$ and $w$ to $C$.  We claim that $|C|\leq k$.  This holds
because we have $3N$ elephants, and each elephant has a ``trunk''-pixel
that must be guarded, but can be guarded only by a camera 
that is entirely within that elephant.  So at least $3N$ cameras
are neither inside a vertex-bar nor
an edge-strip,  and $|C|\leq k$.  Finally we claim that $C$ is a vertex
cover.  Indeed, the strip of any edge $e$ contains a pixel that is not in a 
vertex-bar.  To guard this pixel, we must either use a vertical camera inside 
the edge-strip, or a horizontal camera inside the bar of one of the ends of 
the edge.  Either way, one endpoint of $e$ has been added to $C$.
\end{proof}

Clearly, $P$ can be constructed in polynomial time and so by Lemma~\ref{lem:onlyHSCReduction} NP-hardness follows.  We also notice that $P$ is thin.
NP-hardness of guarding problems in thin polygons (albeit with other models of guards and visibility) have been studied before~\cite{tomas2013,otherESA}.
The NP-hardness holds for both MSC and MSHC; NP-hardness of MSC was known
before \cite{durocher2013} but the constructed polygon was not thin.
We summarize:

\begin{theorem}
\label{thm:mhscHardness}
MSC and MHSC is \textsc{NP}-hard on thin orthogonal polygons with holes.
\end{theorem}

\section{Art-Gallery-Type Results}
\label{sec:artGalleryTheorem}
We now consider art-gallery-type problems for the MSC and MHSC problem; that is, we give tight bounds, depending on $n$, on the number of sliding cameras needed to guard an orthogonal polygon $P$ with $n$ vertices. 

Recall that Aggarwal showed a tight bound $\lfloor (3n+4)/16 \rfloor$ for the number of mobile guards necessary and sufficient to guard $P$~\cite{aggarwal1984}. Closer inspection reveals that the lower bound construction (see Figure~\ref{fig:lowerBoundGeneral}) actually works for sliding cameras, since no two of the $(3n+4)/16$ pixels marked with a cross can be guarded by one camera.

\begin{figure}[t]
\centering
\hspace*{\fill}
\includegraphics[width=.70\textwidth]{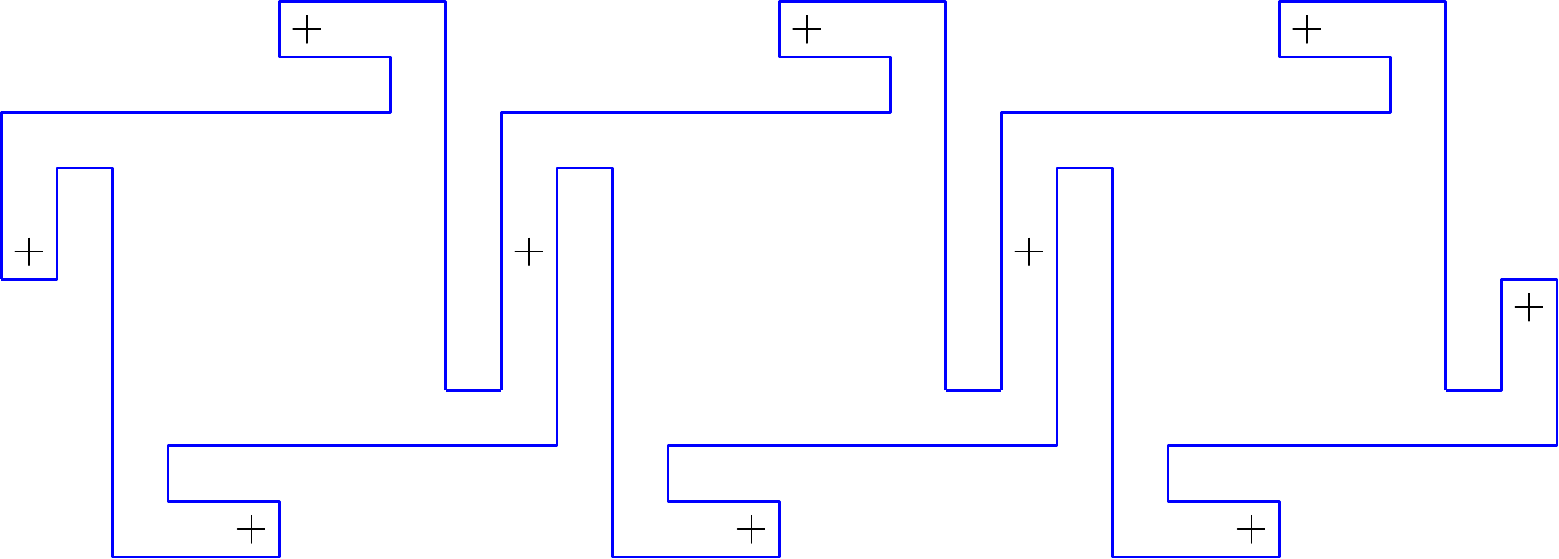}
\hspace*{\fill}
\includegraphics[width=.09\textwidth]{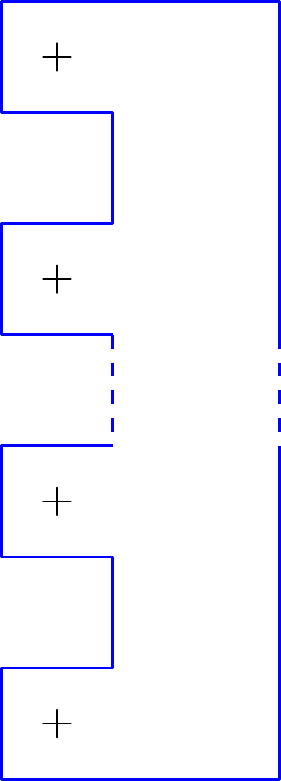}
\hspace*{\fill}
\caption{(Left) A polygon that requires $(3n+4)/16$ cameras.  (Right) A
polygon that requires $n/4$ horizontal cameras.}
\label{fig:lowerBoundGeneral}%
\end{figure}

The upper bound, indeed, also works for sliding guards.  We very briefly review the approach taken in \cite{aggarwal1984}.  
The idea is to guard first a small portion of $P$ using one or two mobile guards, cutting a guarded region out of $P$, and then guarding the rest of $P$ by an induction hypothesis.  There are numerous cases, but in all of them one can establish that indeed a sliding camera would have achieved the same as the mobile guard used. So we have the following result.

\begin{theorem}[Based on~\cite{aggarwal1984}]
\label{thm:bounds}
Given a simple orthogonal polygon $P$ with $n$ vertices, $\lfloor (3n+4)/16 \rfloor$ sliding cameras are always sufficient and sometimes necessary to guard $P$ entirely.
\end{theorem}

Now we give an art-gallery theorem for MHSC.

\begin{theorem}
\label{thm:onlyHSCTheorem}
Given an orthogonal polygon $P$ with $n$ vertices, $\lfloor n/4\rfloor$ horizontal sliding cameras are always sufficient and sometimes necessary to guard $P$ entirely.
\end{theorem}
\begin{proof}
First, observe that that the comb-shaped polygon in Figure~\ref{fig:lowerBoundGeneral} requires $\lfloor n/4\rfloor$ cameras if we only allow horizontal cameras.
For the upper bound, we use O'Rourke's alternate proof 
\cite{orourke1983b} that $\lfloor n/4\rfloor$ stationary point guards 
are always sufficient to guard an orthogonal polygon with $n$ vertices.  He
proves this by showing that any orthogonal polygon can be cut into 
$\lfloor n/4 \rfloor$ orthogonal pieces that have at most one reflex vertex
each.  We claim that each such piece can be guarded with one horizontal
sliding camera.  This is obvious if the piece is a rectangle.  Else (because
the piece is orthogonal and has one reflex corner) it is an $L$-shape, and
by placing the guard in the larger of the horizontal slices we guard
the entire piece.  So $\lfloor n/4 \rfloor$ horizontal cameras are enough
to guard any orthogonal polygon.
\end{proof}

\section{Polygons with Path-Segmentations}
\label{sec:pathPolygons}

Recall that earlier we considered the dual graph of the pixelation of a
polygon $P$.  In this section, we again consider a dual graph, but this
time of one of the segmentations of $P$.  Thus, consider (say) the
vertical segmentation obtained after extending vertical rays from all
reflex vertices.  Interpret the segmentation as a planar graph, and
let $G$ be its weak dual graph obtainined by defining a vertex for
every vertical slice and connecting two slices if and only if they
share (part of) a side.  We say that a polygon {\em has a path-segmentation}
if $G$ is a path for at least one of the two possible segmentations.  
Observe that any monotone orthogonal polygons has a path-segmentation,
but the class is broader than that (e.g. we can ``stack'' multiple
monotone polygons by letting them alternatingly share their rightmost
and leftmost vertical slice.)

In this section, we show that $\lfloor (n+2)/6\rfloor$ sliding cameras are always sufficient and sometimes necessary to guard path polygons.   We need a helper-lemma.

\begin{lemma}
\label{lem:8vertices}
Let $P$ be an orthogonal polygon with up to 8 vertices.  Then $P$ can be guarded with a single sliding camera.
\end{lemma}
\begin{proof}
We already showed (in the proof of Theorem~\ref{thm:onlyHSCTheorem}) that an orthogonal polygon with 6 vertices can be guarded with a  horizontal camera, so assume that $P$ has exactly 8 vertices. Since $P$ is orthogonal, it has exactly two reflex corners.  After possible rotation, we may assume that these two reflex corners do not have the same $x$-coordinate.  Take the vertical decomposition of $P$ (which by the above creates three slices), and consider the middle slice $\sigma$. This slice has exactly one reflex vertex on each of its left and right side, and therefore is incident to the segmentation lines extended from them. So some part (perhaps all) of the left and right side of $\sigma$ is in the interior of $P$.

If there exists some horizontal line segment $\ell$ within $\sigma$ such that on both the left and right side of $\sigma$ it hits the interior of $P$, then  we are done. Namely, extend $\ell$ into a maximal line segment inside $P$;  this gives a horizontal camera that covers all three slices (see Figure~\ref{fig:threeConsecutiveRs}(a)). If there is no such line segment, then on both left and right side of $\sigma$ some part must belong to the boundary of $P$.  But then the slice left  of $\sigma$ has a smaller $y$-range than $\sigma$, as does the slice right of $\sigma$.  Hence a vertical camera inside $\sigma$ guards all of $P$, see Figure~\ref{fig:threeConsecutiveRs}(b). 
\end{proof}

\begin{figure}[ht]
\centering
\hspace*{\fill}
\includegraphics[page=1,height=.20\textwidth]{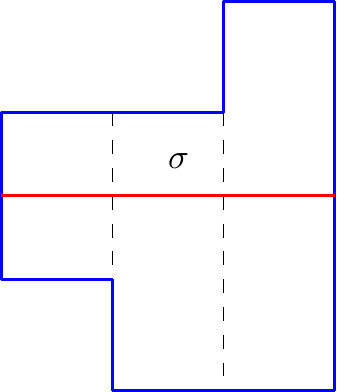}
\hspace*{\fill}
\includegraphics[page=2,height=.20\textwidth]{8sided.pdf}
\hspace*{\fill}
\includegraphics[width=.25\textwidth]{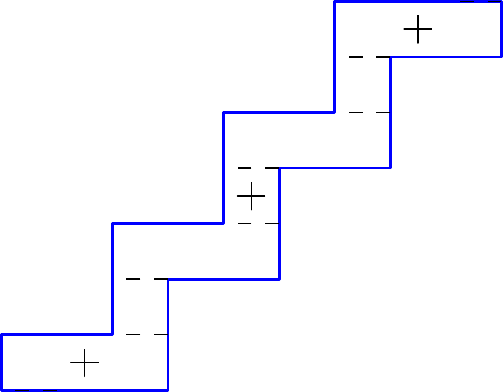}
\hspace*{\fill}
\caption{(Left and middle) Guarding an 8-sided polygon.  (Right)
A polygon with a path-segmentation that requires $(n+2)/6$ cameras.}
\label{fig:threeConsecutiveRs}%
\label{fig:lowerBoundPath}%
\end{figure}

\begin{theorem}
\label{thm:pathPolygons}
Given an orthogonal polygon $P$ with $n$ vertices that has a path segmentation, $\lfloor (n+2)/6\rfloor$ sliding cameras are always sufficient and sometimes necessary to guard $P$ entirely.
\end{theorem}
\begin{proof}
For the lower bound, consider the polygon in
Figure~\ref{fig:lowerBoundPath}(b) (which can be generalized in a natural
way).  We have indicated $3= \lfloor (n+2)/6 \rfloor$ black crosses
with the property that no camera can see two of them.  Hence at least this
many cameras are required. 

For the upper bound, a simple inductive argument suffices.  Clearly the
bound holds for $n=4,6,8$, since we can guard $P$ with one camera, so
assume $n>8$.  We aim to extract a sub-polygon $P'$ of $P$ such that $P'$
has at most 8 vertices and the rest $P-P'$ has again a path segmentation and 
has at most $n-6$ vertices; the result then follows by induction and
Lemma~\ref{lem:8vertices}.  We find $P'$ by removing the first 2 or 3 slices
of the segmentation whose dual is a path (where ``first'' refers to ``at
an end of the path'').   Namely, if the first segmentation line connects
two reflex vertices, then set $P'$ to be the first two slices, else set it
to be the first three slices.  Elementary computation shows that $P'$ has
at most 8 corners, and that $P-P'$ has at least 3 fewer reflex vertices
than $P$, hence $P-P'$ has at most $n-6$ vertices as desired.
\end{proof}

\section{Conclusion}
\label{sec:conclusion}
In this paper, we studied the problem of guarding an orthogonal polygon
with the minimum number of sliding cameras.  We gave the first constant-factor
approximation algorithm for this problem, which works even if the polygon
has holes.  We also showed how to solve the problem optimally if the
polygon is thin and has no holes, and we gave art-gallery-type results
bounding the number of sliding cameras that are always sufficient
and sometimes required.

The most interesting remaining question is whether guarding an orthogonal
polygon with sliding guards is polynomial if the polygon has no holes.
Also, the factor in our $O(1)$-approximation algorithm (which we did not
compute since it is hidden in the machinery of 
\cite{BronnimannG95,clarksonV2007}) is likely large.  Can it be improved?
Even better, could we find a PTAS or is the problem APX-hard?
	
\bibliographystyle{plain}
\bibliography{ref}

\begin{thebibliography}{10}

\bibitem{aggarwal1984}
A.~Aggarwal.
\newblock {\em The Art Gallery Theorem: its variations, applications and
  algorithmic aspects}.
\newblock PhD thesis, Johns Hopkins University, 1984.
\newblock {A summary can be found in \cite{orourke1987}, Chapter 3.}

\bibitem{avis1981}
D.~Avis and G.~T. Toussaint.
\newblock An optimal algorithm for determining the visibility of a polygon from
  an edge.
\newblock {\em {IEEE} Trans. Computers}, 30(12):910--914, 1981.

\bibitem{baker1994}
B.~Baker.
\newblock Approximation algorithms for {NP}-complete problems on planar graphs.
\newblock {\em J. ACM}, 41(1):153--180, 1994.

\bibitem{otherESA}
T.~Biedl and S.~Mehrabi.
\newblock On $r$-guarding orthogonal thin polygons, 2016.
\newblock Manuscript, in preparation.

\bibitem{BronnimannG95}
H.~Br{\"{o}}nnimann and M.~T. Goodrich.
\newblock Almost optimal set covers in finite {VC}-dimension.
\newblock {\em Discrete {\&} Computational Geometry}, 14(4):463--479, 1995.

\bibitem{chvatal1975}
V.~Chv\'{a}tal.
\newblock A combinatorial theorem in plane geometry.
\newblock {\em Journal of Combinatorial Theory, Series B}, 18:39--41, 1975.

\bibitem{clarksonV2007}
K.~L. Clarkson and K.~R. Varadarajan.
\newblock Improved approximation algorithms for geometric set cover.
\newblock {\em Discrete {\&} Computational Geometry}, 37(1):43--58, 2007.

\bibitem{Courcelle1990}
B.~Courcelle.
\newblock The monadic second-order logic of graphs. {I. R}ecognizable sets of
  finite graphs.
\newblock {\em Information and Computation}, 85(1):12--75, 1990.

\bibitem{cygan2015}
M.~Cygan, F.~V. Fomin, L.~Kowalik, D.~Lokshtanov, D.~Mark, M.~Pilipczuk,
  M.~Pilipczuk, and S.~Saurabh.
\newblock {\em Parameterized Algorithms}.
\newblock Springer, Heidelberg, Germany, 2015.

\bibitem{deBergCOCOA2014}
M.~de~Berg, S.~Durocher, and S.~Mehrabi.
\newblock Guarding monotone art galleries with sliding cameras in linear time.
\newblock In {\em Combinatorial Optimization and Applications (COCOA 2014)},
  volume 8881 of {\em LNCS}, pages 113--125, 2014.

\bibitem{durocherLATIN2014}
S.~Durocher, O.~Filtser, R.~Fraser, A.~D. Mehrabi, and S.~Mehrabi.
\newblock A (7/2)-approximation algorithm for guarding orthogonal art galleries
  with sliding cameras.
\newblock In {\em Proceedings of Latin-American Symposium (LATIN 2014)}, volume
  8392 of {\em LNCS}, pages 294--305, 2014.

\bibitem{durocher2013}
S.~Durocher and S.~Mehrabi.
\newblock Guarding orthogonal art galleries using sliding cameras: algorithmic
  and hardness results.
\newblock In {\em Proceedings of Mathematical Foundations of Computer Science
  (MFCS 2013)}, volume 8087 of {\em LNCS}, pages 314--324, 2013.

\bibitem{eidenbenz2001}
S.~Eidenbenz, C.~Stamm, and P.~Widmayer.
\newblock Inapproximability results for guarding polygons and terrains.
\newblock {\em Algorithmica}, 31(1):79--113, 2001.

\bibitem{GareyJ1977}
M.~R. Garey and D.~S. Johnson.
\newblock The rectilinear {S}teiner tree problem in {NP}-complete.
\newblock {\em {SIAM} Journal of Applied Mathematics}, 32:826--834, 1977.

\bibitem{hoffmannFOCS1991}
F.~Hoffmann, M.~Kaufmann, and K.~Kriegel.
\newblock The art gallery theorem for polygons with holes.
\newblock In {\em Proceedings of Foundations of Computer Science (FOCS 1991)},
  pages 39--48, 1991.

\bibitem{jeff1983}
J.~Kahn, M.~M. Klawe, and D.~J. Kleitman.
\newblock Traditional galleries require fewer watchmen.
\newblock {\em SIAM Journal on Algebraic Discrete Methods}, 4(2):194--206,
  1983.

\bibitem{KatzMN2005}
M.~J. Katz, J.~S.~B. Mitchell, and Y.~Nir.
\newblock Orthogonal segment stabbing.
\newblock {\em Comput. Geom.}, 30(2):197--205, 2005.

\bibitem{katz2011}
M.~J. Katz and G.~Morgenstern.
\newblock Guarding orthogonal art galleries with sliding cameras.
\newblock {\em Inter.\ J.\ Comp.\ Geom.\ \& App.}, 21(2):241--250, 2011.

\bibitem{katz2008}
M.~J. Katz and G.~S. Roisman.
\newblock On guarding the vertices of rectilinear domains.
\newblock {\em Comput. Geom.}, 39(3):219--228, 2008.

\bibitem{Kirkpatrick2015}
D.~G. Kirkpatrick.
\newblock An $o(\lg \lg opt)$-approximation algorithm for multi-guarding
  galleries.
\newblock {\em Discrete {\&} Computational Geometry}, 53(2):327--343, 2015.

\bibitem{krohn2013}
E.~Krohn and B.~J. Nilsson.
\newblock Approximate guarding of monotone and rectilinear polygons.
\newblock {\em Algorithmica}, 66(3):564--594, 2013.

\bibitem{lee1986}
D.~T. Lee and Arthur~K. Lin.
\newblock Computational complexity of art gallery problems.
\newblock {\em IEEE Transactions on Information Theory}, 32(2):276--282, 1986.

\bibitem{Lubiw1985}
A.~Lubiw.
\newblock Decomposing polygonal regions into convex quadrilaterals.
\newblock In {\em Proceedings of the ACM Symposium on Computational Geometry
  (SoCG 1985)}, pages 97--106, 1985.

\bibitem{orourke1982}
J.~O'Rourke.
\newblock The complexity of computing minimum convex covers for polygons.
\newblock In {\em 20th Allerton Conf. Commun. Control Comput.}, pages 75--84,
  1982.

\bibitem{orourke1983b}
J.~O{'}Rourke.
\newblock An alternate proof of the rectilinear art gallery theorem.
\newblock {\em J. of Geometry}, 21:118--130, 1983.

\bibitem{orourke1983}
J.~O{'}Rourke.
\newblock Galleries need fewer mobile guards: a variation to {C}hv\'{a}tal's
  theorem.
\newblock {\em Geometriae Dedicata}, 14:273--283, 1983.

\bibitem{orourke1987}
J.~O'Rourke.
\newblock {\em Art Gallery Theorems and Algorithms}.
\newblock The International Series of Monographs on Computer Science. Oxford
  University Press, New York, NY, 1987.

\bibitem{dietmar1995}
D.~Schuchardt and H.~Hecker.
\newblock Two \textsc{NP}-hard art-gallery problems for ortho-polygons.
\newblock {\em Mathematical Logic Quarterly}, 41(2):261--267, 1995.

\bibitem{TamassiaT1986}
R.~Tamassia and I.G. Tollis.
\newblock A unified approach a visibility representation of planar graphs.
\newblock {\em Discrete {\&} Computational Geometry}, 1:321--341, 1986.

\bibitem{tomas2013}
A.P. Tom{\'{a}}s.
\newblock Guarding thin orthogonal polygons is hard.
\newblock In {\em Fundamentals of Computation Theory (FCT 2013)}, volume 8070
  of {\em LNCS}, pages 305--316, 2013.

\end{thebibliography}

\newpage

\end{document}